\documentclass[journal,draftcls,12pt,onecolumn]{IEEEtranTCOM}

\usepackage[numbers]{natbib}

\usepackage{graphicx}
  \DeclareGraphicsExtensions{.pdf,.jpeg,.png,.jpg}
  \DeclareGraphicsExtensions{.eps}
\usepackage{amssymb}
\usepackage{amsfonts}
\usepackage{comment}
\usepackage{amsmath}
\usepackage{times}
\usepackage{authblk}
\usepackage{algorithm}
\usepackage{algorithmic}
\usepackage{booktabs}
\usepackage{subfigure}

\usepackage{caption2}

  \newtheorem{theorem}{Theorem}

  \newcommand{\Pdisk}{basic disk}

\author{Yan~Wang,~\IEEEmembership{Student Member,~IEEE,}
        Xunrui~Yin,
        Xin~Wang,~\IEEEmembership{Member,~IEEE,}
\thanks{Y. Wang, and X. Wang are with the School of Computer Science, Fudan
University (e-mail: \{11110240029, xinw\}@fudan.edu.cn). X. Yin is with the Department of Computer Science, University of Calgary (e-mail: xunyin@ucalgary.ca). He was with Fudan University when the main work was done. \ Y. Wang is also with the School of Software, East China Jiao Tong University.}
}
\title{MDR Codes:  A New Class of RAID-6 Codes with Optimal Rebuilding and Encoding}
\begin{document}
\maketitle
\begin{abstract}
As storage systems grow in size, device failures happen more frequently than ever before.
Given the commodity nature of hard drives employed, a storage system needs to tolerate a certain number of disk failures while maintaining data integrity, and to recover lost data with minimal interference to normal disk I/O operations.
RAID-6, which can tolerate up to two disk failures with the minimum redundancy, is becoming widespread. However, traditional RAID-6 codes suffer from high disk I/O overhead during recovery.
In this paper, we propose a new family of RAID-6 codes, the Minimum Disk I/O Repairable (MDR) codes, which achieve the optimal disk I/O overhead for single failure recoveries. Moreover, we show that MDR codes can be encoded with the minimum number of bit-wise XOR operations.
Simulation results show that MDR codes help to save about half of disk
read operations than traditional RAID-6 codes, and thus can reduce the
recovery time by up to 40\%.
\end{abstract}

\section{Introduction}


To satisfy the storage demand of ``big data'' in data centers, distributed storage systems are typically constructed from a large number of commodity servers and hard drives. As the capacity grows, disk failures happen more frequently than ever before.
RAID-6 systems, which can tolerate two disk failures with the minimum redundancy, have been widely used.



Measurement studies in the literature suggest that single disk failures represent 98.08\% of recoveries \cite{rashmi2013solution}. When there is one or two disk failures, the system has to run at a reduced speed. Hence minimizing the time of single failure recovery is important. Since disk I/O time represents a dominant component in recovery time \cite{diskIO}, the most promising approach to improve the recovery performance is to reduce the amount of data read from each disk \cite{rethinkingOKhan}.

In its general specification, RAID-6 does not impose restrictions on the specific coding method. In fact, one may apply any maximum distance separable (MDS) codes that can tolerate 2 erasures, as exemplified by the rather popular MDS array codes with a row parity block on each row. Many such codes have been designed, such as EVENODD\cite{evenodd}, RDP\cite{RDP}, Liberation Codes\cite{liberationcodes}. The row parity block and data blocks stored in the same row are called a {\em row parity set}. If a single data disk fails, the conventional way of repair is to calculate each block by XORing the blocks remaining on the surviving disks in the row parity set.
\begin{figure}
\begin{center}
  \includegraphics[width=0.9\textwidth]{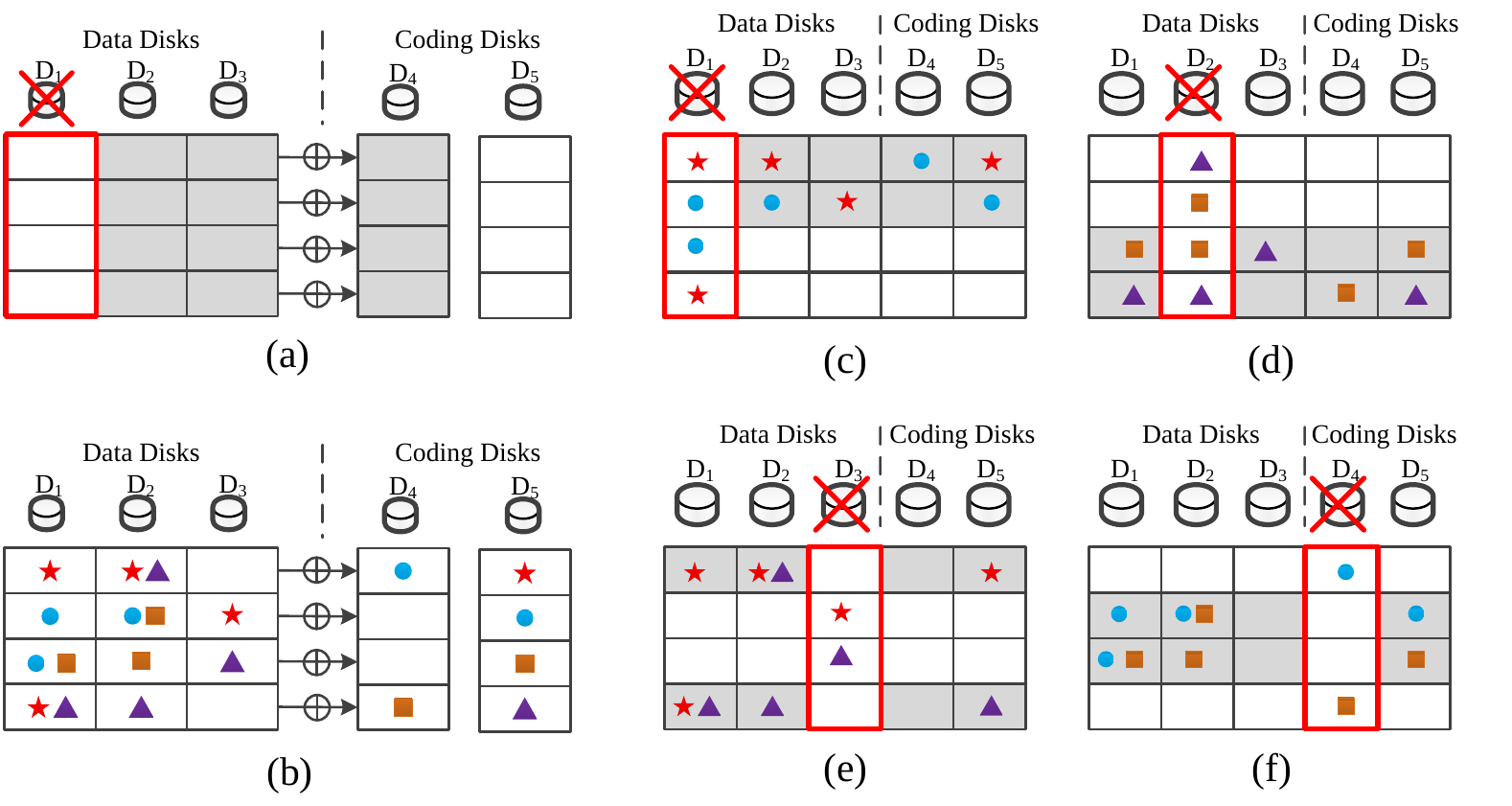}\\
  \caption{An example RAID-6 code with minimum disk I/O for a single failure recovery. When a data disk or the row parity disk fails, only half of the blocks in each surviving disks are read in the recovery.}\label{example}
\end{center}
\end{figure}



However, for the two existing RAID-6 codes, RDP and EVENODD, Xiang {\em et al.} and Wang {\em et al.} showed that if the other coding disk is used in the repair, the failed disk can be recovered by reading $3/4$ blocks from each surviving disk \cite{OptimalRDPXiang}\cite{rebuilding}\cite{hybridxiang2011}.
Furthermore, Tamo {\em et al.} and En Gad {\em et al.} showed that the repair disk I/O can be further reduced if the Q disk is designed carefully \cite{zigzag}\cite{repairGF2}.

Take Fig.~\ref{example} for example.  Disks $D_1$, $D_2$ and $D_3$ are the data disks, each holding 4 un-coded information blocks. Disk $D_4$ is called a P disk, which holds the row parity of data blocks. Disk $D_5$ is called a Q disk. Fig.~\ref{example}(a) shows the conventional way to repair $D_1$, which requires reading 12 blocks. Fig.~\ref{example}(b) shows a RAID-6 code. The blocks on the Q disk are calculated as the parity of the blocks with the same mark labeled in the figure. Fig.~\ref{example}(c)-(f) show the repair strategies for a single failure of each disk except the Q disk, where shaded blocks are read to repair the failed disk. For example, as shown in Fig.~\ref{example}(c), if disk $D_1$ fails, the first two rows of blocks on the surviving disks are read to memory, so that we can calculate the first two blocks of $D_1$ by row parities, and then the last two blocks by the parity sets marked with ``$\circ$'' and ``$\star$'', since all the other blocks in the two parity sets are known. As a result, only 8 blocks are read to repair $D_1$, saving $33.3\%$ disk I/O over conventional repair.

In this work, we study the problem of minimizing disk I/O for the repair of a single disk failure with MDS array codes constructed over $\mathbb{F}_2$, {\em i.e.,} all coded blocks can be generated with bit-wise XOR operations only. Our contributions include the following:

1) We prove exact lower bounds on the minimum disk I/O: at least $(k+1)r/2$ blocks must be read to recover a data disk or the P disk, and at least $kr$ blocks must be read to repair the Q disk, where $k$ is the number of data disks and $r$ is the number of rows in the array. Furthermore, we prove that in the repair of a data disk or the P disk, $r/2$ blocks of the failed disk must be repaired by row parity, in order to achieve the minimum disk I/O.


2) We develop an equivalent condition (Theorem \ref{thm:optimalRepair}) for the optimal repair in RAID-6 codes. With this condition, we find the example repair-optimal code shown in Fig.~\ref{example} and construct the MDR codes, which minimize the repair disk I/O not only for the data disks but also for the coding disks.

3) We show that MDR codes can be encoded with the minimum number of XOR operations. We achieve this by utilizing the intermediate result in the calculation of P disk. To our knowledge, MDR codes represent the first family of codes that minimize both repair disk I/O and computational overhead.

\vspace{2mm}

The rest of this paper is organized as follows. We review related literatures in section \ref{sec:relatedwork} and the specification of RAID-6 codes and basic notations in section \ref{sec:preliminaries}. In section \ref{sec:approach}, we propose a generator matrix approach for studying the minimum disk I/O problem. Along this approach, we prove lower bounds on the minimum disk I/O and develop the equivalent  condition for the optimal repair. In section \ref{sec:construction}, we propose the construction of MDR codes. In section \ref{sec:analysis}, we show how to minimize the computational overhead with MDR codes. We discuss the drawbacks of MDR codes in section \ref{sec:arraysize} and present the simulation results in section \ref{sec:simulation}. Section \ref{sec:conclusion} concludes this paper.

\section{Related Work}\label{sec:relatedwork}

In the design of RAID-6 codes, many researchers focus on minimizing the computational overhead of encoding, updating, and decoding. For example, the EVENODD codes \cite{evenodd} achieve near optimal computational complexity in both encoding and decoding, and the RDP codes \cite{RDP} further improve updating complexity. Plank proposed the Liberation codes \cite{liberationcodes} that are freely available and achieve either optimal or close to optimal in the encoding, updating, and decoding complexity.

Recently, reducing the repair disk I/O attracts more and more attentions. Studies on reducing disk I/O can be divided into two classes: one is to develop clever algorithms for existing RAID-6 codes, and the other is to design new RAID-6 codes.
For the former class, Xiang {\em et al.} \cite{OptimalRDPXiang}\cite{hybridxiang2011} and Wang {\em et al.} \cite{rebuilding} used both parity disks to reduce the disk I/O in single disk failure recovery. They designed efficient recovery algorithms for RDP codes and EVENODD codes. The proposed optimal recovery strategies can reduce approximately $1/4$ disk reads compared with conventional recovery algorithms.
Khan {\em et al.} \cite{rethinkingOKhan}\cite{searchofIOkhan2011} proved that the problem of finding minimum repair disk I/O for a given XOR-based erasure code is NP-hard in general, and Zhu {\em et al.} \cite{speedup2012zhu} proposed a polynomial-time approximation algorithm for this problem. 


The problem of designing new RAID-6 codes to optimize repair disk I/O has been studied in the more general context of optimizing disk I/O for distributed storage systems. Inspired by network coding, Dimakis {\em et al.} \cite{ncfordss} proved a lower bound on the minimum bandwidth consumption in the recovery.
As the amount of data transmitted is always no more than the amount of data read, the repair disk I/O is at least the minimum repair bandwidth. Therefore, Dimakis' lower bound on the latter directly implies a lower bound on the former, which implies that reading at least $(k+1)r/2$ blocks is necessary for the repair in RAID-6 codes with $k$ data disks and $r$ rows in the array. According to the study of minimum bandwidth with exact repairs \cite{exactRepairD2K3}, it is impossible for a $(k>4, r=2)$ RAID-6 code to achieve this lower bound in the repair of every single disk.

However, Tamo {\em et al.} and En Gad {\em et al.} recently showed that the bound $(k+1)r/2$ is achievable for the repair of a data disk.
Specifically, Tamo {\em et al.} proposed the Zigzag MDS array codes that minimize the repair disk I/O \cite{zigzag}. Their codes require coding over a field of size at least 3 and achieve optimal update as well. Furthermore, Zigzag codes have strip size $r=2^{k-1}$, and they proved that this strip size is optimal for all systematic, update-optimal and repair-optimal MDS codes. 
En Gad {\em et al.} \cite{repairGF2} also proposed a family of RAID-6 array codes over $\mathbb{F}_2$ that achieve the optimal repair disk I/O for a data disk recovery.

Our work differs from the above in the following aspects. First, they only optimized disk I/O for the repair of {\em data disks}, while we consider the repair of {\em every disk}. Both of their codes require reading $kr$ blocks to repair the row parity disk, but MDR codes require reading only $(k+1)r/2$ blocks. We further prove that the minimum disk I/O to repair the Q disk is at least $kr$ in RAID-6 codes with a row parity disk.
Second, MDR codes also minimize the computational overhead, which is not considered in their works.
To the best of our knowledge, MDR codes are the first that minimize repair disk I/O and computational overhead at the same time.
Third, we proposed a generic approach for constructing repair-optimal RAID-6 codes from an initial code satisfying certain conditions, which can be found by computer search.

Compared with Zigzag codes, a drawback of MDR codes is that we trade-off update disk I/O for restricting coding operations to over $\mathbb{F}_2$, the same as in En Gad's codes. However, recent reports show that there are many archive-style storage systems where update operations are rare. For example, in Windows Azure \cite{azure}, the storage system is used in an append-only way. We also trade-off the strip size for the optimal encoding complexity --- the strip size of MDR codes is twice as much as in the Zigzag codes and En Gad's codes.

\begin{table*}[!htb]
	\centering
  \caption{Comparison between repair-optimal codes.}
    \begin{tabular}{|l|c|p{2cm}|c|c|c|c|}
    \hline
      & {\bfseries Field Size} & {\bfseries \#Disk Repairs Improved} & {\bfseries Strip Size} & {\bfseries  Disk I/O in Update} & {\bfseries  Encoding Complexity} \\
    \hline
    Zigzag codes \cite{zigzag} & $\geq 3$  &  $k$ & $2^{k-1}$ & 2 & --- \\
    \hline
    En Gad's codes\cite{repairGF2} & 2  & $k$ & $2^{k-1}$ & $1/2 \cdot \lfloor k/2 \rfloor +2$ & $k+k/2 \cdot \lfloor k/2 \rfloor$ \\
    \hline
    MDR codes & 2  & $k+1$ & $2^{k}$ & $(k+7)/4 $ & $k-1$ \\
    \hline
    \end{tabular}
    \label{tab:comparison}
\end{table*}

We summarize the comparison in Table~\ref{tab:comparison}, where ``\#Disk Repairs Improved'' refers to the number of disks that can be repaired with reading $(k+1)r/2$ blocks, ``Disk I/O in Update'' refers to the average number of parity blocks changed in the update of a data block, and ``Encoding Complexity'' refers to the number of XORs to compute each block of the Q disk. For encoding complexity, Zigzag codes require $k-1$ additions and up to $k$ multiplications over a field of size at least $3$. We note that the optimal encoding complexity of MDR codes is achieved under the condition that the P disk is computed at the same time. The repair disk I/O, encoding complexity, and update disk I/O of MDR codes are analyzed in Sec.~\ref{sec:construction}, Sec.~\ref{sec:analysis} and Sec.~\ref{sec:arraysize}, respectively.

\section{Preliminaries and Notations}
\label{sec:preliminaries}
\subsection{Erasure Codes and RAID-6 Specification}
Erasure codes ensure data reliability by encoding a message of $k$ symbols into $n$ symbols, so that the message can be recovered even if some symbols are lost. An optimal erasure code can tolerate the loss of any $m = n-k$ symbols. We can find such an optimal erasure code in linear codes, which is also called an $(n,k)$-MDS code. Compared with replication, MDS codes are storage efficient, since replication requires to store $mk$ symbols instead of $k+m$ symbols to provide the same reliability against $m$ erasures.

In the specification of RAID-6, there are $k+2$ storage nodes each holding the same amount of data, and up to two node failures can be tolerated.
A RAID-6 code can therefore be viewed as an $(n=k+2,k)$-MDS code. RAID-6 requires $k$ data disks to store original information, and hence is a {\em systematic code}. Two coding disks further store coded data.
In this work, we use $D_1, D_2, \cdots, D_k$ to denote the $k$ data disks, and $D_{k+1}, D_{k+2}$ to denote the two coding disks, which are called P disk and Q disk, respectively.

\subsection{Parity Array Coding Technique}

In a parity array code, data stored on each disk is grouped into $r$ blocks of equal length, which are called a strip. Bit-wise XOR is applied to these blocks to generate parity blocks.
Blocks in disks $D_1, D_2, \cdots, D_{k+2}$ are typically arranged into an array of $r$ rows and $k+2$ columns.
Let $d_{i,j}, (1\leq i\leq k+2, 1\leq j \leq r)$ denote the $j$-th block in disk $D_i$, and $d$ be the column vector $[d_{1,1}\  d_{1,2}\  \cdots \  d_{k+2,r}]^T$.

\begin{figure}
\centering
  \includegraphics[width=0.48\textwidth]{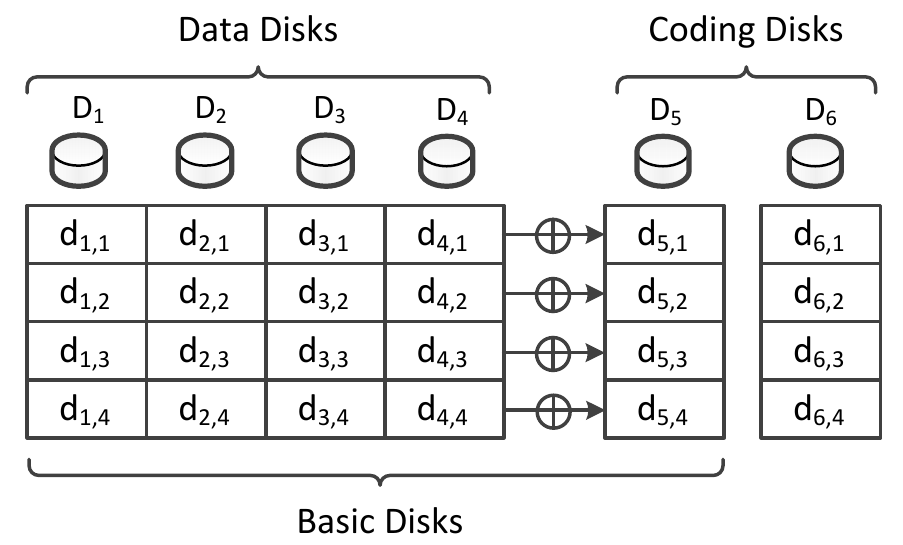}\\
  \caption{Block arrangement in RAID-6 codes with a row parity disk.}\label{fig:raidcode}
\end{figure}


Most implementations of RAID-6 array codes use the first coding disk $D_{k+1}$ as a row parity, {\em i.e.}, $\forall j=1,\cdots,r$,
$
d_{k+1,j} = d_{1,j} + d_{2,j} + \cdots + d_{k,j}
$,
where addition is over the finite field $\mathbb{F}_2$. This makes it easily extendable from RAID-5 by simply adding another coding disk. Fig.~\ref{fig:raidcode} illustrates the general idea of how parity blocks are calculated in such codes. Due to the similarity between the row parity disk $D_{k+1}$ and the data disks $D_1, D_2, \cdots, D_k$, we call them the {\em \Pdisk{}s}.


\subsection{Notations}
For a positive integer $n$, we use $[n]$ to represent the set $\{1, 2, \cdots, n\}$. For an $m$-by-$n$ matrix $A$, a row index set $R \subset [m]$ and a column index set $C \subset [n]$, we use $A|_{R,C}$ to denote the sub-matrix of $A$ induced by the $R$ rows and $C$ columns. We refer to a RAID-6 code supporting $k$ data disks with $r$ rows in the array as  a $(k,r)$ RAID-6 code. For simplicity, we assume $r$ is even.

\section{The Generator Matrix Approach}\label{sec:approach}
In this section, we use the generator matrix to formulate the problem of minimizing repair disk I/O, and develop an equivalent condition for optimal repair.


Let $d_i$ denote the column vector of the $r$ blocks $[d_{i,1}\ d_{i,2} \ \cdots \ d_{i,r}]^T$ in disk $D_i$. The generator matrix is illustrated in Fig.~\ref{fig:matrixvector}, where the shaded elements of the matrix are ones and the other elements are zeros.  As we consider RAID-6 codes with a row parity disk, we have $d_{k+1} = d_1 + d_2 + \cdots + d_k$. Therefore, to design a RAID-6 code, we only need to specify how the coding disk $D_{k+2}$ is coded. According to the generator matrix, $d_{k+2}$ can be written as
\[
d_{k+2} = A_1 d_1 + A_2 d_2 + \cdots + A_k d_k
\] where $A_1, A_2, \cdots, A_k$ are square matrices of size $r$.

\begin{figure}
\centering
  \includegraphics[width=0.48\textwidth]{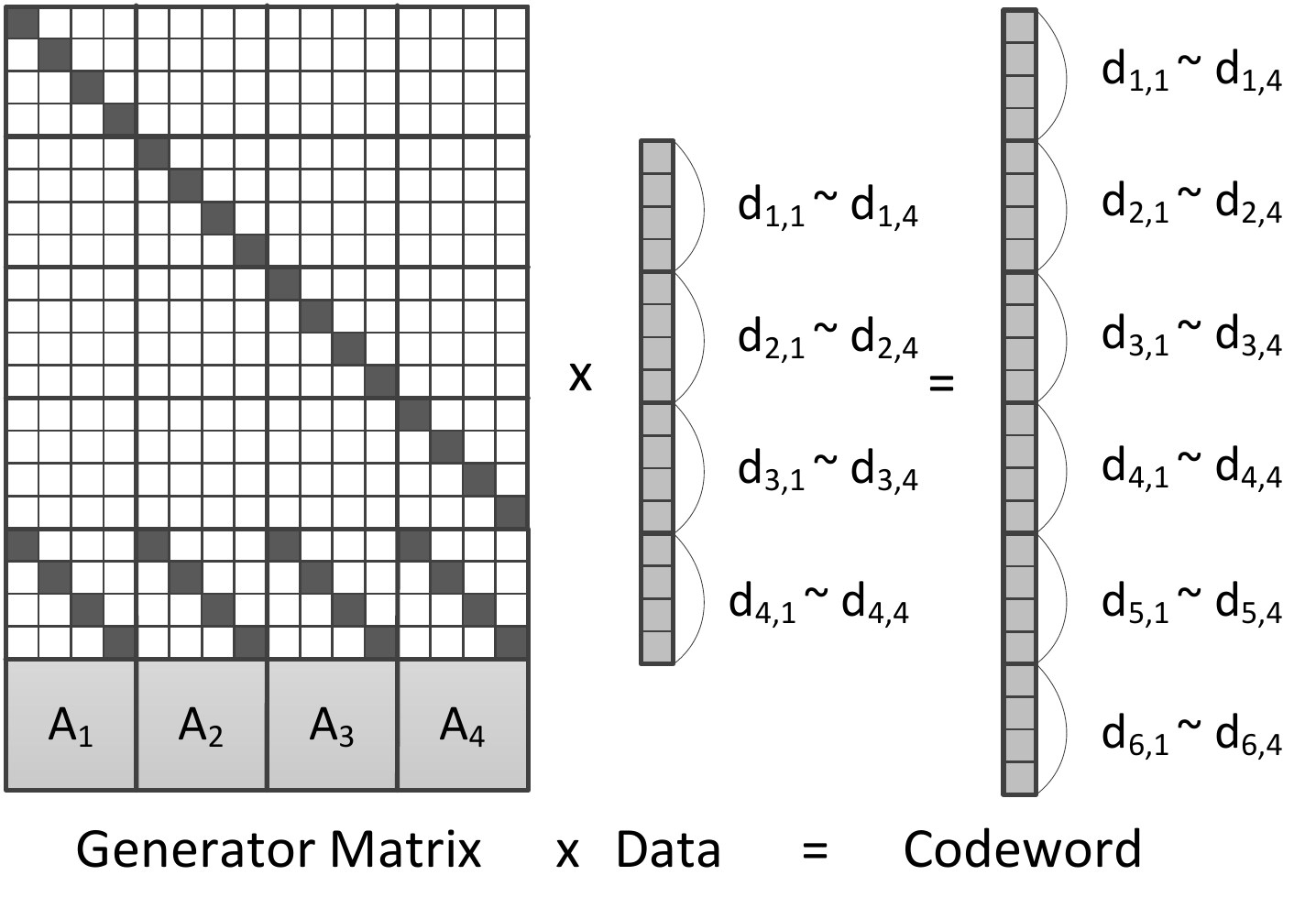}\\
  \caption{A RAID-6 code with a row parity disk is uniquely determined by its generator sub-matrices $A_1, A_2, \cdots, A_k$. }\label{fig:matrixvector}
\end{figure}

We call $A_1, A_2, \cdots, A_k$ the {\em generator sub-matrices}. An XOR-based RAID-6 code with a row parity disk is uniquely determined by its generator sub-matrices.

RAID-6 requires that the system can be reconstructed from any two disk failures. According to the study of Blaum and Roth \cite{lowestdensityMDS}, the code described by matrices $A_1, A_2, \cdots, A_k$ satisfies such RAID-6 specification if and only if $A_1, A_2, \cdots, A_k$ satisfy the following conditions, which we refer to as the MDS property:
\begin{itemize}
\item $A_i$ is non-singular for all $i\in [k]$
\item $A_i + A_j$ is non-singular for all $i,j\in [k], \  i \neq j$
\end{itemize}

\subsection{A Matrix Representation of the Minimum Disk I/O Problem}
With two parity disks, each data block $d_{i,j}$ can be represented by the sum of other data blocks in multiple ways.
This makes it difficult to find the minimum repair disk I/O. In fact, Khan \cite{rethinkingOKhan} showed that solving this problem for an XOR-based code is NP-hard in general. In this subsection, we rewrite this problem in terms of the generator sub-matrices $A_1, A_2, \cdots, A_k$.

The problem of finding minimum disk I/O is essentially the same as representing the lost data blocks by a set of surviving data blocks of minimum size. A key observation is that all representations are based on the following equation in terms of parity-check matrix $H$
\begin{eqnarray}
H d
=
\left[
\begin{array}{cccccc}
I & I & \cdots & I & I & 0\\
A_1 & A_2 & \cdots & A_k & 0 & I
\end{array}
\right]
\left[
\begin{array}{c}
d_1\\
d_2\\
\vdots\\
d_{k+2}
\end{array}
\right]
= 0
\label{eqn:basic}
\end{eqnarray}
Each row of $H$ describes an equation that can be interpreted as a representation of each block involved in it. Take the case of $r=2, k=2$ as an example, the first row of the parity-check matrix
$
\left[
\begin{array}{cccccc}
I_{2\times 2} & I_{2\times 2} & I_{2\times 2} & 0\\
A_1 & A_2 &  0 & I_{2\times 2}
\end{array}
\right]
$
is $[1\  0\ \    1\  0\  \   1\  0\  \   0\  0]$, which means $d_{1,1} + d_{2,1} + d_{3,1} = 0$ and equivalently, any one of $d_{1,1}, d_{2,1} , d_{3,1}$ can be represented as the sum of the other two. In fact, each representation is equivalent to an equation that can be derived as a linear combination of the rows in equation  (\ref{eqn:basic}). This observation leads to the following theorem.
\begin{theorem}
Let $\hat{N}(A)$ denote the number of non-zero columns in matrix $A$, then the minimum disk I/O to recover the row parity disk $D_{k+1}$ equals
\[
\min_{\textrm{$X$}} \hat{N}([I + XA_1\quad  I+XA_2\quad \cdots \quad I+XA_k\quad  X])
\]
where $X$ is a square matrix of size $r$. The minimum disk I/O to recover the coding disk $D_{k+2}$ equals
\[
\min_{\textrm{$X$}} \hat{N}([X + A_1\quad  X+A_2\quad \cdots \quad X+A_k\quad  X])
\]\label{thm:minIOfml}
\end{theorem}
\vspace{-6mm}
\begin{IEEEproof}
According to our previous analysis, every representation is a linear combination of rows in equation (\ref{eqn:basic}), which can be described as:
$
\left[ v_1 \  v_2 \   \cdots \  v_{2r} \right] H d = 0
$. To recover the disk $D_{k+1}$ is equivalent to represent the lost data $d_{k+1}$ by $r$ equations. Group the equations into the form $V_{r \times 2r} H d = 0$ and rewrite $V_{r\times 2r} = [Y\  Z]$, where $Y,Z$ are square matrices of size $r$:
\[
[Y\ Z] \left[
\begin{array}{cccccc}
I & I & \cdots & I & I & 0\\
A_1 & A_2 & \cdots & A_k & 0 & I
\end{array}
\right] \left[
\begin{array}{c}
d_1\\
d_2\\
\vdots\\
d_{k+2}
\end{array}
\right]
= \sum_{i=1}^{k} (Y+ZA_i)d_i + Y d_{k+1} + Z d_{k+2}= 0
\]
In order to solve $d_{k+1}$, its coefficient matrix $Y$ must be invertible. Block $d_{i,j}$ is used in the recovery if and only if the $j$-th column of the coefficient matrix of $d_i$ contains some 1's. Therefore, the total number of blocks used in the recovery equals the number of non-zero columns in the matrix
$
[Y+ZA_1\quad  Y+ZA_2 \quad  \cdots \quad  Y+ZA_k \quad  Z]
$.
For a column vector $z$, $Y^{-1}z$ is a zero vector if and only if $z=0$, since $Y^{-1}$ is of full rank. Therefore, left-multiplying the matrix with $Y^{-1}$ does not change its number of zero columns, and the case of recovering $D_{k+1}$ is proved with $X=Y^{-1}Z$.
For the case of $D_{k+2}$, $Z$ must be non-singular, and the statement can be derived by left-multiplying the matrix with $Z^{-1}$ and letting $X=Z^{-1}Y$.
\end{IEEEproof}

To calculate the minimum disk I/O for repairing a data disk $D_i, i\in [k]$, we may first treat disk $D_i$ as the row parity disk by eliminating sub-matrix $A_i$ in the parity-check matrix $H$, and then applying Theorem \ref{thm:minIOfml}. For example, consider the case of recovering $D_1$, we may left-multiply equation (\ref{eqn:basic}) with
$\left[\begin{array}{cc}
I & 0\\
A_1 & I
\end{array}\right]
$, so that the parity-check matrix $H$ is transformed into:
\[
\left[
\begin{array}{cccccc}
I & I & \cdots & I & I & 0\\
0   & A_2 + A_1 & \cdots & A_k + A_1 & A_1 & I
\end{array}
\right]
\]
Then the minimum disk I/O can be calculated in a similar way as in Theorem \ref{thm:minIOfml}.

\subsection{A Lower Bound on the Minimum Recovery Disk I/O}\label{sec:lowerbound}
In this subsection, we use the generator matrix approach to prove the lower bound on the disk I/O of repairing a single disk failure in any RAID-6 codes with a row parity disk. Note that Dimakis {\em et al.} \cite{ncfordss} have proved an achievable lower bound on the minimum repair bandwidth for functional repair, which implies that the number of blocks read is at least $(k+1)r/2$. Theorem \ref{thm:nBlocksFromEachDisk} strengthens this result for exact repair in RAID-6 codes in two aspects: 1) for the repair of the Q disk, we prove that the minimum disk I/O is at least $kr$; 2) Dimakis' theorem assumes each disk transmits the same amount of information. We drop this assumption and prove that each surviving disk must read at least $r/2$ blocks to repair a basic disk.

\begin{theorem}
The minimum disk I/O to recover a basic disk $D_i, i\in [k+1],$ is at least $(k+1)r /2$.
Further, the amount of data read from each surviving disk must be no less than $r /2$.
To recover the coding disk $D_{k+2}$, the minimum disk I/O is at least $kr$.
\label{thm:nBlocksFromEachDisk}
\end{theorem}
\begin{IEEEproof}
Firstly, consider the case of repairing disk $D_{k+1}$. Let $X$ be a matrix that maximizes the number of zero columns in matrix
$
[I+XA_1 \quad I+XA_2 \quad \cdots \quad I+XA_k \quad X]
$.
According to the MDS property, for any $i,j\leq k, i\neq j$, matrix
$\left[\begin{array}{cc} I & I\\ A_i & A_j\end{array}\right]$
is of full rank, which implies matrix
$\left[\begin{array}{cc} I+XA_i & I+XA_j\\ A_i & A_j\end{array}\right]$
is non-singular, and therefore, the rank of matrix $[I+XA_i $ $I+XA_j]$ must be $r$. So the total number of zero columns in $[I+XA_i\quad I+XA_j]$ is at most $r$. Similarly, as matrix
$\left[\begin{array}{cc} I & 0\\ A_i & I\end{array}\right]$
has full rank, we can conclude that the total number of zero columns in $[I+XA_i\quad X]$ is no more than $r$. Let $z_i, i\in [k]$ denote the number of zero columns in $I+XA_i$, and $z_{k+1}$ denote the number of zero columns
in $X$. From the following optimization problem:
\begin{eqnarray*}
\max & \sum_{i=1}^{k+1} z_i & \\
\textrm{subject to:} & z_i + z_j \leq r & \forall i,j\leq k+1, i\neq j
\end{eqnarray*}
we can see that the maximum total number of zero columns is $(k+1)r/2$ for $k\geq 2$, and the optimal value is achieved only with $z_1 = z_2 = \cdots = z_{k+1} = r/2$. According to Theorem \ref{thm:minIOfml}, the minimum disk I/O is at least $(k+1)r/2$ and is only achieved by reading $r - r/2 = r/2$ blocks from each surviving disk.

Secondly, for the case of repairing a data disk $D_i, i\in [k]$, we may consider disk $D_i$ as the row parity disk with generator sub-matrices $\{A_i, A_j + A_i \ |\  i,j\in [k], i\neq j\}$. The same result can be concluded in a similar way.

Finally, consider the case of repairing the coding disk $D_{k+2}$.
We can see that the indices of zero columns in $X+A_1, X+A_2, \cdots, X+A_k, X$ can not be the same, since if $X+A_i$ and $X+A_j$ has a zero column at the same position, the matrix $X+A_i + X+A_j = A_i+A_j$ has a zero column, which conflicts with the MDS property that $A_i+A_j$ is nonsingular. Similarly, if $X + A_i$ and $X$ has a zero column at the same position, we will obtain $A_i$ is singular, which conflicts with the MDS property. Therefore, there are at most $r$ zero columns in matrix
\[[X+A_1 \quad X+A_2 \quad \cdots \quad X+A_k \quad X]\]
According to Theorem \ref{thm:minIOfml}, the minimum disk I/O to repair $D_{k+2}$ is at least $kr$.
\end{IEEEproof}

\subsection{An Equivalent Condition for Optimal Repair}
From the above analysis, we conclude that the minimum disk I/O to repair coding disk $D_{k+2}$ is $kr$, which is achieved by reading all data blocks. Thus we only need to consider the case of repairing a basic disk.

In particular, a repair strategy is represented by the set of blocks read from each surviving disk. Let $C_{i,j}$ denote the rows index set of blocks read from disk $D_i$ to repair disk $D_j$. For example, if $d_{2,1}, d_{2,3}, d_{2,4}$ are read from disk $D_2$ to repair disk $D_1$, then $C_{2,1} = \{1,3,4\}$. According to Theorem \ref{thm:nBlocksFromEachDisk}, $|C_{i,j}| = r/2$. The following theorem shows that we actually do not need so many sets to describe a repair strategy. In an optimal strategy, the row indices for each basic disk must be the same.

\begin{theorem}
In the repair of a basic disk $D_j, j\in [k+1]$, the minimum disk I/O is achieved only by reading the same rows of the surviving basic disks, i.e., $C_{1,j} = C_{2,j} = \cdots = C_{j-1,j} = C_{j+1,j} = \cdots = C_{k+1,j}$. Namely, $r/2$ lost blocks must be recovered by row parity.
\label{thm:samerow}
\end{theorem}
\begin{IEEEproof}
Without loss of generality, suppose we are repairing disk $D_1$. We need to prove that $C_{2,1} = C_{3,1} = \cdots = C_{k+1,1}$. We may regard each block $d_{i,j}$ as a random variable and consider their entropy. For simplicity, assume each block contains only 1bit, {\em i.e.},  $H(d_{i,j}) = 1$.

In the repair of $D_1$, let $S$ be the set of blocks read from $D_2, D_3, \cdots, D_{k+1}$. Let $T$ be the set of blocks read from $D_{k+2}$, and $R_i (i\in [r])$ be the $i$-th row of basic disks $D_1, D_2, \cdots, D_{k+1}$. Consider the entropy of blocks $D_1 \cup S$. As $D_{k+1}$ is the row parity disk, blocks from different rows are independent. In particular, we have $H(D_1, S) = \sum_{i=1}^{r} H(R_i \cap (D_1 \cup S))$. For each row $i$, $H(R_i \cap (D_1 \cup S)) = |R_i \cap(D_1 \cup S)| - 1$ only if $D_1 \cup S$ contains the entire row $R_i$. Otherwise, $H(R_i \cap (D_1 \cup S)) = |R_i \cap(D_1 \cup S)|$. Noting that $|\cap_{l=2}^{k+1} C_{l,1}|$ indicates the number of rows fully contained in $D_1\cup S$, we have
\[
H(D_1, S) = \sum_{i=1}^{r} |R_i \cap(D_1 \cup S)| - |\cap_{l=2}^{k+1} C_{l,1}| = (k+2)r/2 - |\cap_{l=2}^{k+1} C_{l,1}|
\]
As $|C_{2,1}| = |C_{3,1}| = \cdots = |C_{k+1,1}| = r/2$, if the row indices $C_{2,1}, C_{3,1}, \cdots, C_{k+1,1}$ are not the same, $|\cap_{l=2}^{k+1} C_{l,1}| < r/2$ and $H(D_1, S) > (k+1)r/2$. As $D_1$ can be reconstructed with $S\cup T$,
\begin{eqnarray*}
H(S, T) = H(D_1, S, T) \geq H(D_1, S) > (k+1)r /2
\end{eqnarray*}
which is a contradiction since there are only $(k+1)r/2$ blocks in $S\cup T$. Therefore, the row indices $C_{2,1}, C_{3,1}, \cdots, C_{k+1,1}$ must be the same.
\end{IEEEproof}

Therefore, in the recovery of a basic disk $D_j$, $j\in [k+1]$, we may use $C_j \subset [r]$ to denote the index set of blocks read from the surviving basic disks and $R_j \subset [r]$ to denote the index set of blocks read from the Q disk. The following theorem rewrites the condition of optimal repair in terms of generator sub-matrices.

\begin{theorem}
For a $(k\geq 2,r)$ RAID-6 code described by generator sub-matrices $A_1, A_2, \cdots, A_k$, each basic disk $D_i, i\in[k+1]$  can be repaired by a repair strategy $R_i, C_i, |R_i| = |C_i| = r/2$ if and only if there exists a matrix $B_{k+1}$, such that the series of matrices $B_j = A_j + B_{k+1}, j\in [k]$ and $B_{k+1}$ satisfy the following conditions:

1) For each $i \in [k+1]$, the sub-matrix $B_i |_{R_i,\overline{C_i}}$ is non-singular, and

2) For any $i, j\in [k+1], i \neq j$, $B_j |_{R_i,\overline{C_i}} = 0$.
\label{thm:optimalRepair}
\end{theorem}
\begin{proof}
Please refer to the appendix for the proof.
\end{proof}

According to the definition of $B_i, i\in [k]$ in Theorem \ref{thm:optimalRepair}, the generator sub-matrices $A_i, i\in [k]$ can be derived as $A_i = B_i + B_{k+1}$. Therefore, a RAID-6 code can also be described by the $B_i, i\in [k+1]$ matrices. Note that the series of $B_i$ matrices is not unique for a RAID-6 code. Theorem \ref{thm:optimalRepair} states that if a RAID-6 code can be repaired with minimum disk I/O, it must has a series of $B_i$ matrices satisfying the two conditions 1) and 2).

Fig.~\ref{fig:optimalRepair} illustrates this idea. We represent the example RAID-6 code shown in Fig.~\ref{example} with $B_1, B_2, B_3, B_4$. The repair strategy for disk $D_1$ is reading the first two rows, thus $R_1=C_1 = \{1,2\}$ and  $\overline{C_1}=\{3,4\}$. As the sub-matrix $B_1 |_{R_1,\overline{C_1}}$ is non-singular and the corresponding sub-matrices of $B_2,B_3$ and $B_4$ are zero, $D_1$ can be repaired with optimal disk I/O.

\begin{figure}[h]
\centering
  \includegraphics[width=0.43\textwidth]{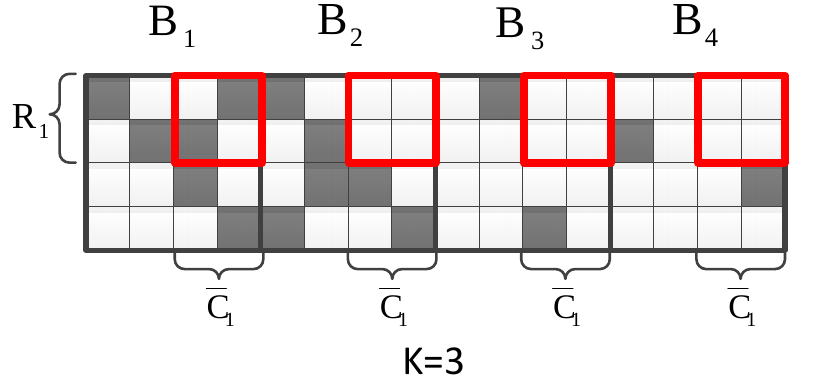}\\
  \caption{An example illustrating the equivalent condition of optimal repair.} 
  \label{fig:optimalRepair}
\end{figure}

Recall that $d_{k+1} = \sum_{i=1}^k d_i$, we have
\[
d_{k+2} = \sum_{i=1}^k A_i d_i =  \sum_{i=1}^k B_i d_i + B_{k+1}\sum_{i=1}^k  d_i = \sum_{i=1}^{k+1} B_i d_i
\]
As for $i,j\in [k]$, $A_i + A_j =  B_i + B_j$, the MDS property is equivalent to
\begin{itemize}
\item $ B_i + B_j $ is non-singular for all $i,j \in [k+1], i\neq j$
\end{itemize}

\section{Construction of MDR Codes}\label{sec:construction}
In this section, we propose the recursive construction approach that derives MDR codes. The approach starts with a repair-optimal RAID-6 code with small $k$ and $r$ that satisfies the following condition:

P1) For each $i\in [k-1]$,  $B_i$ is non-singular;

P2) $R_i = C_i, \, \forall i\in [k+1]$.

We will construct a $(k'=k+1, 2r)$ RAID-6 code with optimal repair disk I/O which also satisfies conditions (P1) and (P2). Denote the new code with matrices $B'_i, i\in [k'+1]$, and repair strategy $R'_i, C'_i, i \in [k'+1]$.
\begin{eqnarray*}
B'_i  &=& \left\{\begin{array}{ll}
\left[\begin{array}{cc} B_i+B_{k+1} & 0 \\ 0 & B_i+B_{k+1} \end{array}\right] &  \textrm{if $i\in [k'-1]$}\\
\vspace{-4mm}&\\
\left[\begin{array}{cc} 0 & I_{r \times r} \\ 0 & 0\end{array}\right] &  \textrm{if $i=k'$ }\\
\vspace{-4mm}&\\
\left[\begin{array}{cc} 0 & 0 \\ I_{r \times r} & 0\end{array}\right] &  \textrm{if $i=k'+1$}
\end{array}\right.
\\
R'_i  = C'_i &=& \left\{\begin{array}{ll}
\{x | x \in R_i \textrm{ or } x-r \in R_i\} &  \textrm{if $i\in [k'-1]$}\\
[2pt] [r] & \textrm{if $i=k'$} \\
[2pt] [ 2r ] \backslash [r] &  \textrm{if $i=k'+1$}
\end{array}\right.
\end{eqnarray*}
\begin{theorem}
The constructed $(k'=k+1, 2r)$ code is a RAID-6 code with optimal repair disk I/O and satisfies conditions (P1) and (P2).
\end{theorem}
\begin{proof}
We first prove the constructed code $B'_i, i\in [k'+1],$ is a RAID-6 code by showing that $B'_i + B'_j$ is non-singular for any $i,j \in [k'+1], i\neq j$. As the code $B_i$ is a RAID-6 code, $B_i + B_j$ has full rank for any $i,j \in [k+1], i\neq j$. The possible results of $B'_i + B'_j$ are $\left[\begin{array}{cc}B_i + B_j & 0 \\ 0 & B_i+B_j\end{array}\right]$, $\left[\begin{array}{cc}B_i + B_{k+1} & I_{r\times r} \\ 0 & B_i+B_{k+1}\end{array}\right]$, $\left[\begin{array}{cc}B_i + B_{k+1} & 0 \\ I_{r\times r} & B_i+B_{k+1}\end{array}\right]$, $\left[\begin{array}{cc}0 & I_{r\times r} \\ I_{r\times r} & 0\end{array}\right]$, which are all non-singular. {$\rule{0ex}{16pt}$} Therefore, the constructed code is a RAID-6 code.

Second, we use Theorem \ref{thm:optimalRepair} to show that the constructed code can recover from single disk failure with repair strategy $R'_i, C'_i, i\in [k'+1]$. According to the hypothesis, $R_i, C_i, i\in [k+1]$ is the optimal recover strategy for code $B_i$, we have $B_{k+1}|_{R_i, \overline{C_i}} = 0$ for $i\in [k]$ by Theorem \ref{thm:optimalRepair}. For $i,j \in [k]$,
\[
B'_j |_{R'_i, \overline{C'_i}} =
\left[\begin{array}{cc}B_j|_{R_i, \overline{C_i}} + B_{k+1}|_{R_i, \overline{C_i}} & 0 \\ 0 & B_j|_{R_i, \overline{C_i}} + B_{k+1}|_{R_i, \overline{C_i}}\end{array}\right]
= \left[\begin{array}{cc}B_j|_{R_i, \overline{C_i}} & 0 \\ 0 & B_j|_{R_i, \overline{C_i}} \end{array}\right]
\]
which is non-singular if $i=j$, and zero matrix otherwise. For $i \in [k'-1]$, as $R_i = C_i$, $I|_{R_i, \overline{C_i}}=0$. Thus $B'_{k+1}|_{R'_i, \overline{C'_i}} = \left[\begin{array}{cc}0 & I|_{R_i, \overline{C_i}}  \\ 0 & 0 \end{array}\right] = 0$. For the same reason, we have $B'_{k'+1}|_{R'_i, \overline{C'_i}} = 0$. Thus disk $D_i, i\in [k'-1]$ can be recovered by the strategy $R'_i,C'_i$. For the case of $i=k'$ and $i=k'+1$, we can see that $R'_i, C'_i$ satisfy the optimal repair conditions of Theorem \ref{thm:optimalRepair}.

Finally,  condition (P2) is directly satisfied from the construction, and condition (P1) is also satisfied, since for $i\in [k'-1]$, $B'_i = \left[\begin{array}{cc} B_i+B_{k+1} & 0 \\ 0 & B_i+B_{k+1} \end{array}\right]$ is non-singular according to the previous analysis.
\end{proof}

\begin{figure}
\centering
  \includegraphics[width=0.9\textwidth]{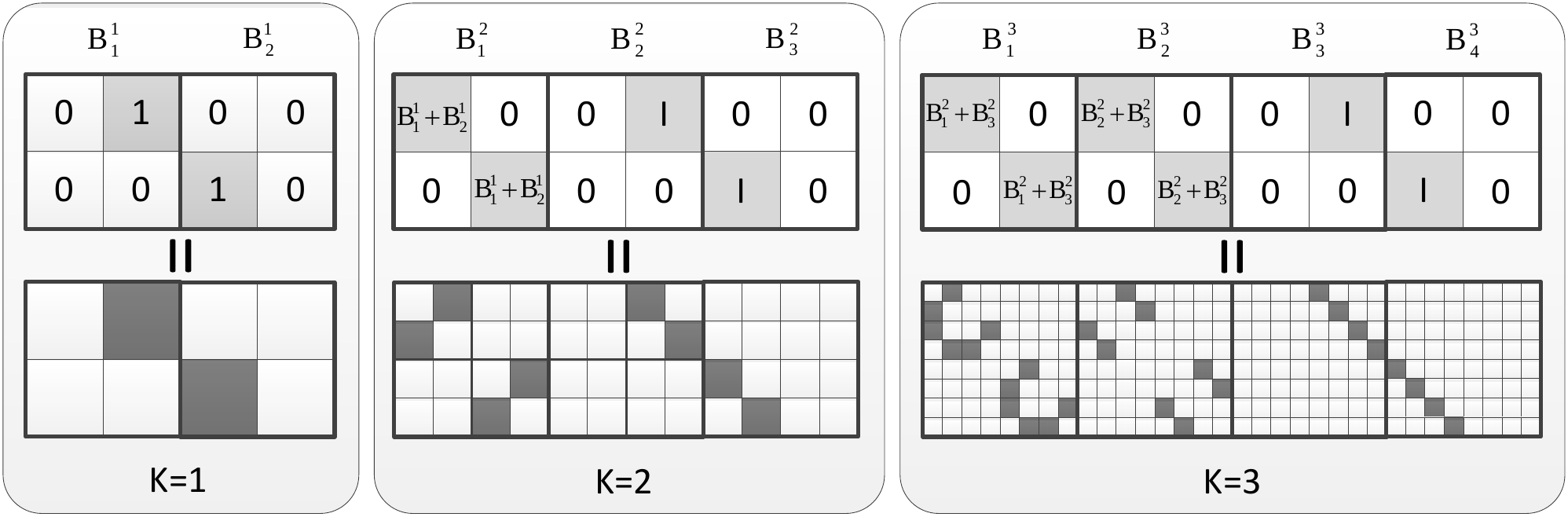}\\
  \caption{The construction of MDR codes.}\label{fig:fsr}
\end{figure}

In order to generate MDR codes with this approach, we need an initial repair-optimal RAID-6 code satisfying (P1) and (P2). The following $(k=1,r=2)$ RAID-6 code can be applied:
\begin{eqnarray*}
B_1 = \left[\begin{array}{cc}0 & 1\\ 0 & 0\end{array}\right] &\ ,\ & B_2 = \left[\begin{array}{cc}0 & 0\\ 1 & 0\end{array}\right]\\
R_1 = C_1 = \{1\} &\  ,\  & R_2 = C_2 = \{2\}
\end{eqnarray*}
We can verify that this code satisfies the requirement of our approach. The resulting MDR codes are shown in Fig.~\ref{fig:fsr} for the case $k=1,2,3$.

\section{Achieving the Minimum Encoding Overhead with MDR Codes}\label{sec:analysis}
Besides disk I/O overhead, another important metric of a RAID-6 code is coding complexity. In this section, we will show that we can achieve the minimum encoding complexity with MDR codes.

\subsection{Encoding Complexity}
\label{sec:encoding}
The encoding complexity considers the number of XORs to generate the coding disks P, Q. As the row parity disk P can be directly computed using $k-1$ XORs for each block, we only need to consider computing the Q disk. A direct way to compute the Q disk is to calculate $d_{k+2} = \sum_{i=1}^{k}A_id_i$. But with the knowledge of row parity disk, we may calculate the Q disk with fewer XOR operations.

We accomplish this in a recursive way. Let $y_t$ denote the number of XORs to calculate the Q disk in the case $k=t$. Let $B_i, i\in [t+1]$ represent the $(t-1, 2^{t-1})$ MDR code and $B'_i, i\in [t]$ represent the $(t, 2^{t})$ MDR code constructed from $B_i$. Let $d'_i, i\in[t+2]$ denote the blocks in disk $D_i$,
\begin{eqnarray*}
d'_{t+2} & = &
\sum_{i=1}^{t-1}\left[\begin{array}{cc}
B_i & 0 \\
0 & B_i \\
\end{array}
\right] d'_i + \left[\begin{array}{cc}
B_{t} & 0 \\
0 & B_{t} \\
\end{array}
\right] \sum_{i=1}^{t-1} d'_i + \left[\begin{array}{cc}
0 & I \\
0 & 0 \\
\end{array}
\right] d'_t + \left[\begin{array}{cc}
0 & 0 \\
I & 0 \\
\end{array}
\right] d'_{t+1}
\end{eqnarray*}
Carefully checking the above formula, we can see that calculating the upper half of the first two terms is equivalent to calculating the Q disk in the $(t-1, 2^{t-1})$ MDR code, with the upper half of $d'_i, i\in [t-1]$ as the data blocks. According to the induction hypothesis, if we know the corresponding P disk of the $(t-1, 2^{t-1})$ MDR code, we can calculate the first two terms with $2y_{t-1}$ XORs. Fortunately, we are able to acquire this information, {\em i.e.,} the value of $\sum_{i=1}^{t-1}d'_i$, in the calculation of the P disk of the $(t, 2^t)$ MDR code. Thus, we have
\[
y_{t} = 2 y_{t-1} + 2^{t-1} + 2^{t-1}
\]
Note that for the initial case $k=1$, blocks in the P disk and Q disk are replications of data blocks, which means $y_1 = 0$. Solving this recursive equation, we obtain $y_{t} = (t-1) 2^t$. As there are $r = 2^t$ blocks in the coding disk, the average number of XORs to compute one coded block is $k-1$, which is optimal \cite{liberationcodes}.

\begin{figure}
\centering
  \includegraphics[width=0.42\textwidth]{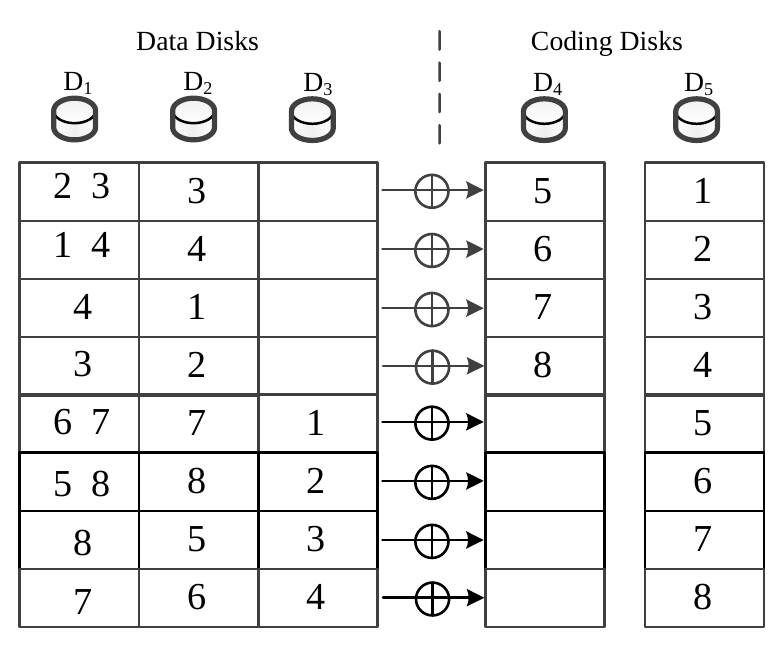}\\
  \caption{The MDR code with $k=3$. The Q disk is calculated as the parity of blocks containing the same number.}\label{fig:mdrcode}
\end{figure}

\noindent{\bf Example.} Consider the $(k=3,r=8)$ MDR code shown in Fig.~\ref{fig:mdrcode}. We use this example to explain how the Q disk can be calculated with $k-1=2$ XORs for each block. For the first block in the Q disk, we may directly compute it with $d_{5,1} = d_{1,2} + d_{2,3} + d_{3,5}$, which takes 2 XORs. For the third block in the Q disk, $d_{5,3} = d_{1,1} + d_{1,4} + d_{2,1} + d_{3,7}$. As we may catch the intermediate result $d_{1,1} + d_{2,1}$ in the computation of P disk, we can see that $d_{5,3}$ can also be computed with 2 XORs.

%


\subsection{Recovery Complexity}
The recovery complexity counts the average number of XORs to regenerate a failed block in a single disk failure recovery.
\begin{theorem}
With MDR codes, a failed basic disk can be recovered by $k-1$ XORs for computing each lost block.
\end{theorem}
\begin{IEEEproof}
Please refer to the appendix for the proof.
\end{IEEEproof}

Note that if only the Q disk fails, we are currently unable to recover it with $k-1$ XORs, since applying the method in Sec.~\ref{sec:encoding} requires computing the P disk at the same time, which results in $2(k-1)$ XORs for rebuilding each block of the Q disk.

\section{Discussions}\label{sec:arraysize}
We have seen in previous sections that MDR codes achieve optimal repair disk I/O and encoding complexity. To achieve this, however, we need a large update disk I/O and a large strip size, which will be discussed in this section.

\subsection{Update Disk I/O}
Update disk I/O refers to the number of parity blocks that are affected by changing the content of a data block. As the number may vary for updating different data blocks, we focus on the average value here. In terms of the parity-check matrix $H$, the update disk I/O equals the average number of ones in each column of $H$.

For a ($k, 2^{k}$) MDR code, let $x_k$ denote the average number of ones in each column of matrices $A_i = B_i + B_{k+1}, i\in [k]$. Then $k2^{k}x_k$ is the total number of 1's in the matrix $[A_1 \  A_2 \ \cdots \ A_k]$. According to the construction in Sec.~\ref{sec:construction}, we have
\[
k2^{k}x_{k} = 2(k-1) 2^{k-1} x_{k-1} + (k+1) 2^{k-1}
\]
And for the case of $k=1$, we have $x_1=1$. Solving this recursive equation,
\begin{eqnarray*}
k x_{k} & = & (k-1) x_{k-1} + \frac{k + 1}{2}\\
& = & (2-1)x_1 + \sum_{i=2}^{k} \frac{i+1}{2} = \frac{1}{2}( k(k+1)/2 + k)
\end{eqnarray*}
we obtain $x_k = (k+3)/4$ for $k\geq 2$. Note that $x_k$ is actually the average number of parity blocks in the Q disk that are affected by the updating. As there is a row parity block affected as well, the update disk I/O of MDR codes is $x_k + 1 = (k+7)/4$.

\subsection{Strip Size}
Here we use the term ``strip size'' to denote the number of rows in a RAID-6 array code. According our construction, the MDR codes have strip size $2^{k}$. The other two repair-optimal codes, Zigzag codes \cite{zigzag} and En Gad's codes\cite{repairGF2}, have strip size $2^{k-1}$. Tamo {\em et al.} proved that, in order to achieve both optimal repair and optimal update at the same time, the minimum strip size is $2^{k-1}$ \cite{zigzag}. If the assumption of optimal updating is dropped, however, the minimum strip size for optimal repair is unknown yet. Inspired by the works \cite{zigzag,repairGF2}, we carried out a brute-force search (based on Theorem \ref{thm:optimalRepair}) for the minimum strip size and verified that the minimum strip size is indeed $2^{k-1}$ for $2\leq k\leq 4$. For $k=5$, the strip size is no less than $2^{k-1}-2 =14$. We conjecture that for a repair-optimal RAID-6 code over $\mathbb{F}_2$, the minimum strip size is at least exponential to the number of data disks $k$.

In practice, the impact of large strip sizes is that we need a large memory to cache these blocks during the repair process. We note that, for a $(k,r=2^{k})$ MDR code, we may carry out the repair by caching only $r/2 + 2$ blocks, because $r/2$ lost blocks are recovered by row parity, which can be computed one by one, with one block for the intermediate result and another for the current block read from the disk. Assume the block size is set to 512 Bytes \cite{EMC}, the memory overhead of MDR codes with 16 data disks is about $512$B $\times 2^{16}/ 2 =  16$MB, which is acceptable.

\section{Simulation}\label{sec:simulation}
\subsection{Simulation Setup}
To evaluate the performance of MDR codes, we use the popular disk simulator Disksim \cite{disksim} to simulate the recovery process of a RAID-6 system.

In the simulation, we set up 10 disks in all, which are connected to an interleaved I/O bus. Our simulation module acts as an I/O driver of a RAID-6 system, which directly generates I/O requests and handles the request completion interrupts from each individual disk. The logical layout is rotated among stripes of the disk array. As I/O tasks on different disks can be executed simultaneously, we pipeline the recovery of sequential stripes, {\em i.e.}, we write the recovered blocks of the last stripe at the same time of reading blocks of the current stripe.
To simulate the workload during an online repair, we also generate random I/O requests to surviving disks.

We implement three different recovery algorithms: the conventional recovery algorithm that uses the row parity, the RDOR \cite{OptimalRDPXiang} recovery algorithm for the RDP code and our recovery algorithm for the MDR codes. Performance of the conventional recovery algorithm is used as a benchmark, so that we use the performance ratio of the other recovery algorithms to the conventional algorithm to measure the improvements.

Two metrics are tested. One is recovery time, and the other is the average access time of the surviving disks, which is measured as the sum of the response time of each I/O request and thus reflects the load on each disk.

\subsection{Impact of Strip Size}
We vary the block size from 512B to 8KB, so that the strip size changes from 32KB to 512KB. Simulations are carried out with 8 disks in the array, and both online repair mode and offline repair mode are tested.

Fig.~\ref{fig:resultatgroup} shows the average access time with different strip sizes. As disk access time of the recovery process is not affected by I/O requests of the other application process, we only show the result of online repair mode. We can see that, compared with the conventional repair algorithm, the RDOR algorithm reduces access time to about 80\% $\sim$ 85\%, and our codes further reduce the access time to 65\% $\sim$ 70\%. As strip size grows, access times decrease for all the three recovery algorithms due to sequential reads, but the improvement ratio does not exhibit any characterizable trend.

\begin{figure}
\centering
\subfigure[]{\includegraphics[width=5.64cm,height=4.7cm]{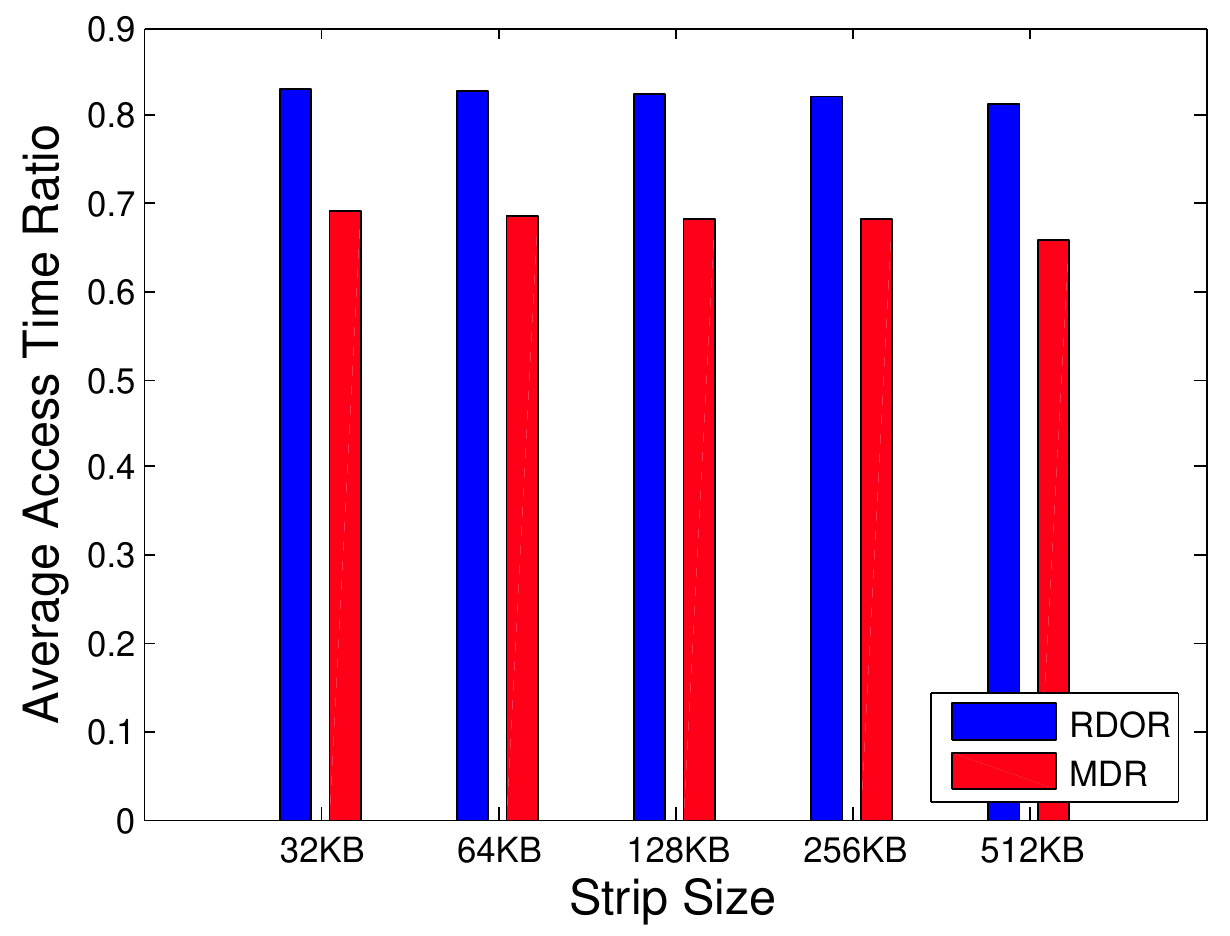}\label{fig:resultatgroup}}
\subfigure[]{\includegraphics[width=5.64cm,height=4.7cm]{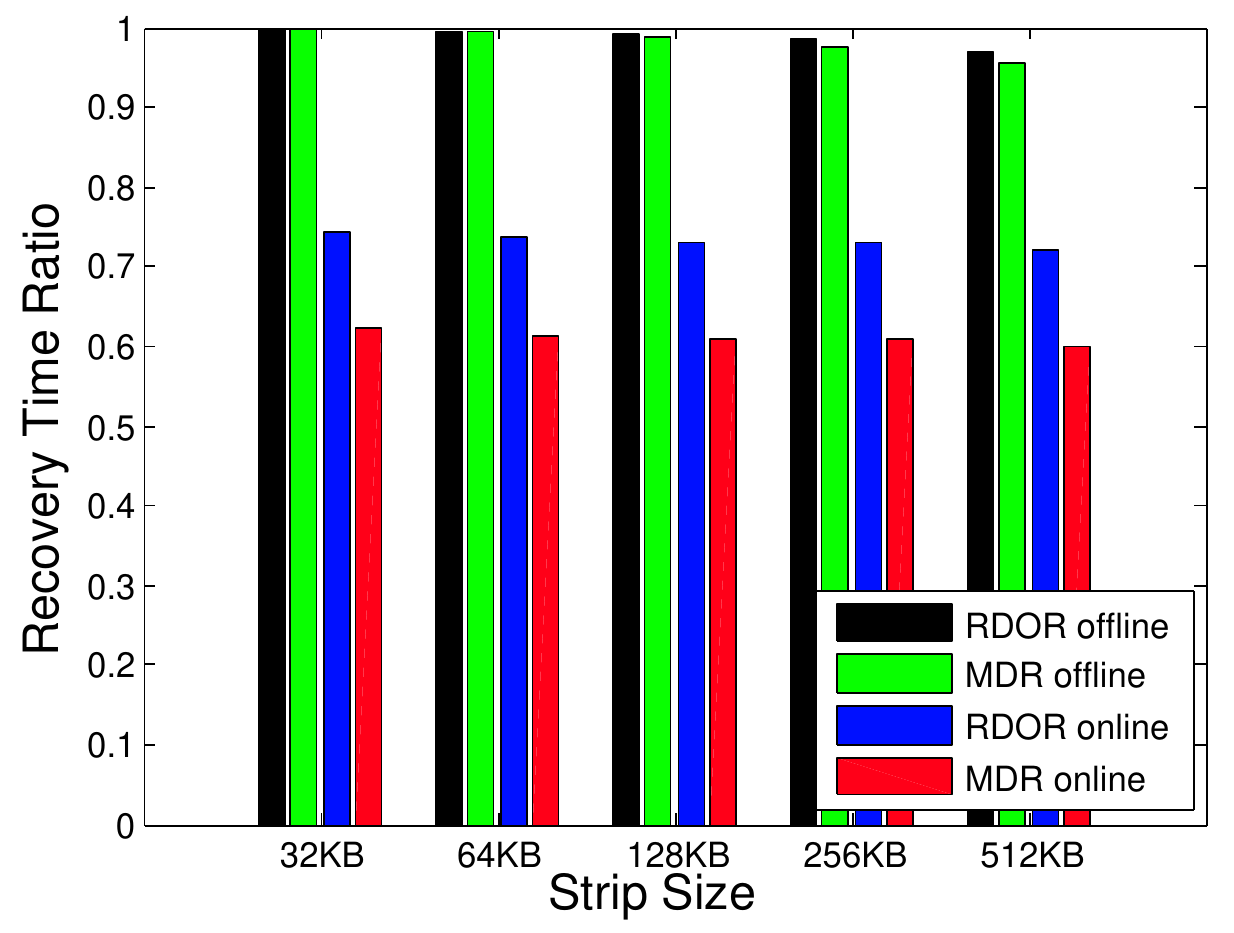}\label{fig:resultgroup}} \\
\subfigure[]{\includegraphics[width=5.64cm,height=4.7cm]{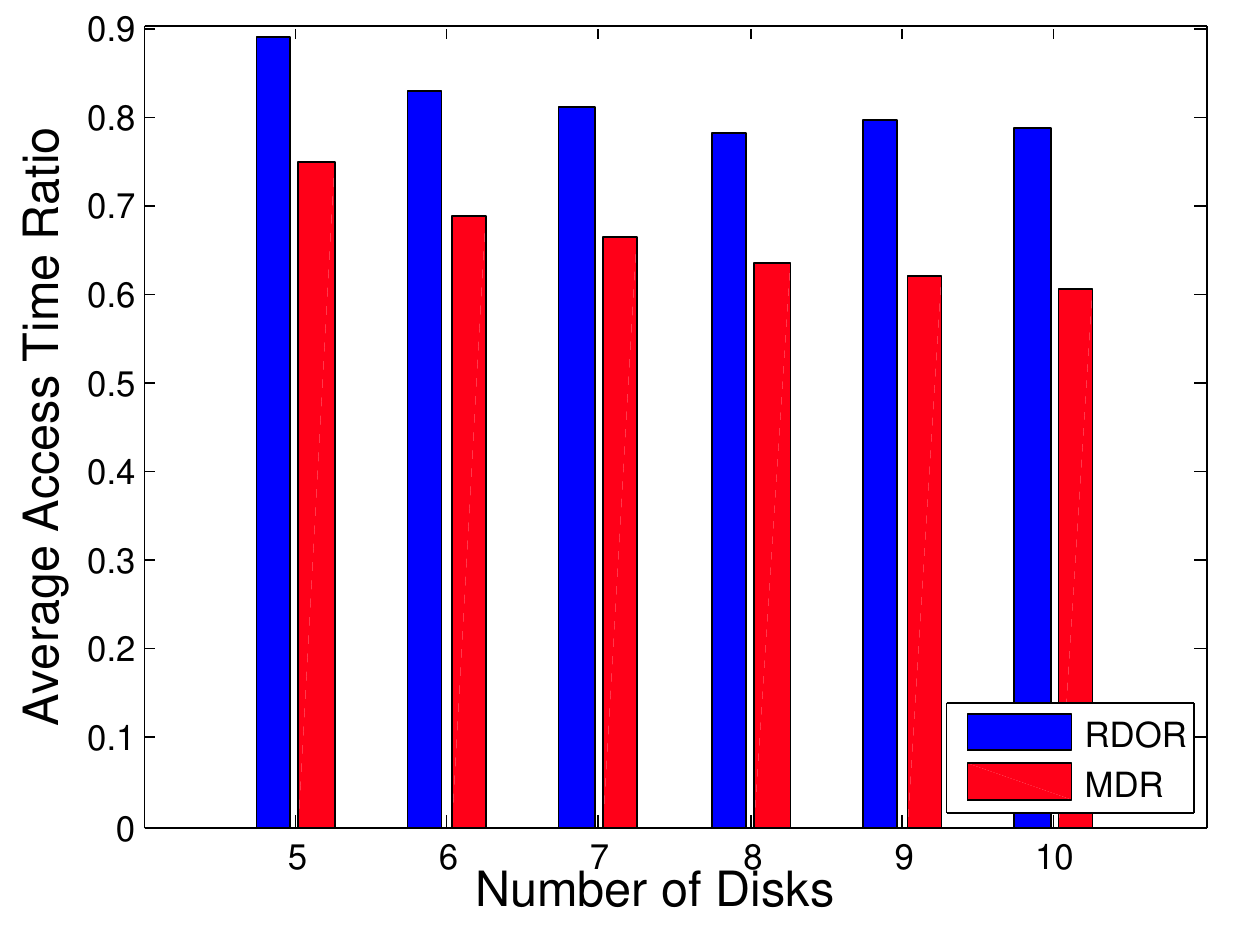}\label{fig:resultatk}}
\subfigure[]{\includegraphics[width=5.64cm,height=4.7cm]{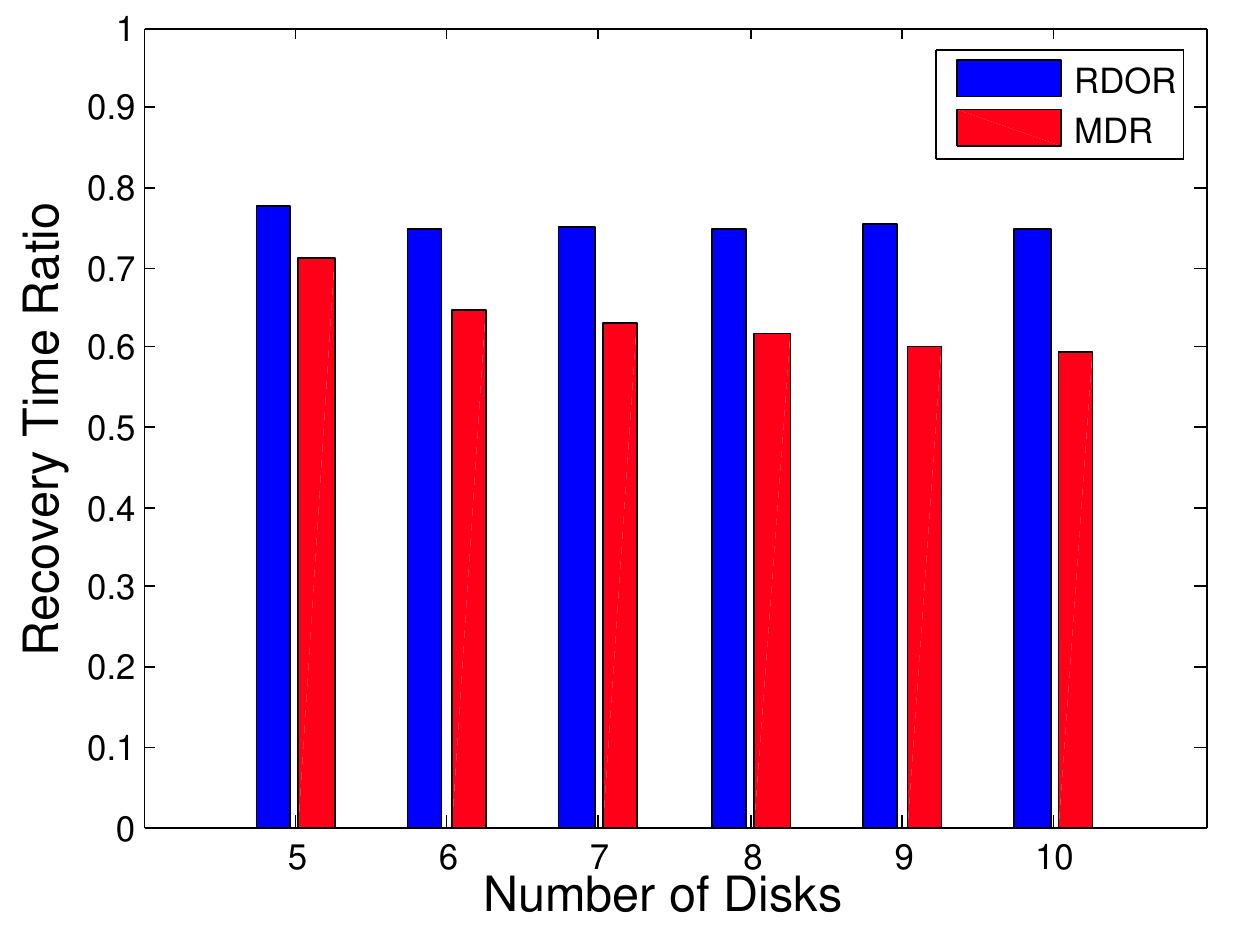}\label{fig:resultk}}
\caption{Simulation results.}
\end{figure}


Fig.~\ref{fig:resultgroup} shows the recovery time with different strip sizes. When the recovery process is carried out in the off-line mode, both the RDOR and our recovery algorithm can hardly reduce the recovery time. This is because we pipeline the disk read and write, and multiple disk reads at different disks can be executed simultaneously. Hence disk write becomes the bottleneck. The recovery time ratio decreases slightly as strip size grows. This is because we can not pipeline for recovering the first strip, whose recovery time is reduced by reading less blocks from the surviving disks.

\subsection{Impact of the Number of Disks}
We increase the number of disks from 5 to 10 to evaluate its impact on the performance of these recovery methods.

Fig.~\ref{fig:resultatk} shows the online average access time with different number of disks.
As the number of disks grows, the access time of RDOR does not show a tendency, while the access time of our codes decreases. This is because the ratio of read disk reads by RDOR depends on the smallest prime that is larger than the number of systematic disks $k$ and at the proper primes, the ratio approaches 75\% as $k$ increases. On the other hand, our ratio approaches 50\% as $k$ increases. For the recovery time,  Fig.~\ref{fig:resultk} shows a similar result.
%

\section{Conclusion}\label{sec:conclusion}
We studied the problem of minimizing disk I/O for every single disk repair in a RAID-6 system, including not only the repair of a data disk but also the repair of a coding disk. We solved this problem by proving a lower bound on the minimum repair disk I/O and constructing the MDR codes that achieve this bound. We also showed that the MDR codes achieve the minimum computational overhead among all RAID-6 codes when the P disk is calculated at the same time. The construction approach is a generic one, which may be used to generate new repair-optimal RAID-6 codes with different initial codes. The main drawback of MDR codes is that its strip size is $2^k$, which is twice as much as that of Zigzag codes and En Gad's codes. Inspired by these codes, we recently modified the construction method for  MDR codes and constructed a new family of RAID-6 codes with strip size $2^{k-1}$, which achieve the optimal repair disk I/O in the repair of every disk and can be encoded with $k$ XORs per each coding block.

We implemented our codes and tested their performance through simulations. Results show that our codes can efficiently reduce the reading overhead of surviving disks to about half of that in the conventional way, and the total recovery time can be reduced by up to 40\%.

\section*{Acknowledgments}
We thank the anonymous JSAC reviewers for their feedback and comments. We thank Jun Li, Zongpeng Li, Ziyu Shao, Yucheng He, and Xiao Ma for helping us improve the writing.

{
\bibliographystyle{IEEEtran}
\bibliography{IEEEabrv,sigcommref}
}

\appendix
\subsection{Proof of Theorem 4}
\begin{IEEEproof}[Proof]
For the ``if'' part: Consider the repair of disk $D_i, i\in [k+1]$. As we read the $C_i$ rows of surviving basic disk, we can rebuild the $C_i$ blocks of $D_i$ with row parity. For a $r$-dimension column vector $d$ and a index set $C\subset [d]$, we use $d|_C$ to denote the column vector formed by the $C$ elements of $d$.  According to the definition of $B_j, j\in[k+1]$, we have $d_{k+2} = \sum_{j=1}^{k+1} B_j d_j$. Consider the $R_i$ rows of this equation:
\[
d_{k+2} |_{R_i} \ =\  \sum_{j=1}^{k+1} ( B_j |_{R_i, C_i} d_j|_{C_i} + B_j |_{R_i, \overline{C_i}} d_j|_{\overline{C_i}})
\   = \   B_i |_{R_i, \overline{C_i}} d_i|_{\overline{C_i}} + \sum_{j=1}^{k+1} B_j |_{R_i, C_i} d_j|_{C_i}
\]
where the second equality is because $B_j |_{R_i, \overline{C_i}} = 0$ for $j\neq i$. As the blocks $d_{k+2}|_{R_i}$, $d_j|_{C_i}$ are read (or calculated by row parity) and $B_i |_{R_i, \overline{C_i}}$ is non-singular, we can see that blocks $d_i|_{\overline{C_i}}$ can be recovered.

For the ``only if'' part: Assume there is a $(k,r)$ repair-optimal RAID-6 code described by the generator sub-matrices $A_1, A_2, \cdots, A_k$ with repair strategy $R_i, C_i, i\in[k+1]$. Let $A_{k+1} = 0$ for simplicity.
As disk $D_i$ can be recovered by reading $R_i$ blocks from the Q disk and $C_i$ blocks from each surviving basic disk, $d_i$ can be represented by these read blocks. Similar to the proof of Theorem \ref{thm:minIOfml}, as every representation is derived as a linear combination of rows of equation (\ref{eqn:basic}), we can use matrix $[Y Z]$ to denote the linear combination such that $Y+ZA_i = I$ and $d_i$ is solved through the equation
$
\sum_{j=1}^{k+1} (Y+ZA_j)d_j + Z d_{k+2} = 0
$.
Let $E$ denote the universal set $\{1,2,\cdots, r \}$. In this equation, the coefficients of unread blocks must be zero, which means $(Y+ZA_j)|_{E,\overline{C_i}} = 0$ for $j \in [k+1],  j\neq i$, and $Z|_{E,\overline{R_i}} = 0$. Subsequently, for $j\in[k+1], j\neq i$
\begin{eqnarray*}
(Y+ZA_j)|_{E,\overline{C_i}}&=& Y|_{E,\overline{C_i}} + Z|_{E,\overline{R_i}} A_j|_{\overline{R_i},\overline{C_i}} + Z|_{E,R_i} A_j|_{R_i,\overline{C_i}} \\
& =  &Y|_{E,\overline{C_i}} + Z|_{E,R_i} A_j|_{R_i,\overline{C_i}} = 0
\end{eqnarray*}
Therefore, $Z|_{E,R_i} A_j|_{R_i,\overline{C_i}} = Y|_{E,\overline{C_i}} = (I+ZA_i)|_{E,\overline{C_i}}$.
Consider the $\overline{C_i}$ rows of the this equation,
\begin{eqnarray*}
Z|_{\overline{C_i},R_i} A_j|_{R_i,\overline{C_i}} & = & I|_{\overline{C_i},\overline{C}} + Z|_{\overline{C_i},E} A_i|_{E,\overline{C_i}} \\
& = & I_{r/2 \times r/2} + Z|_{\overline{C_i},R_i} A_i|_{R_i,\overline{C_i}} + Z|_{\overline{C_i},\overline{R_i}} A_i|_{\overline{R_i},\overline{C_i}} \\
& = & I_{r/2 \times r/2} + Z|_{\overline{C_i},R_i} A_i|_{R_i,\overline{C_i}}
\end{eqnarray*}
where the last equality is because $Z|_{E,\overline{R_i}} = 0$. Namely, we obtain that for $j\in [k+1], j\neq i$,
\[
Z|_{\overline{C_i},R_i} (A_j+A_i)|_{R_i,\overline{C_i}} = I_{r/2 \times r/2}
\]
which means the sub-matrix $(A_j+A_i)|_{R_i,\overline{C_i}}$ are invertible and the same for $j\in [k+1], j\neq i$. Note that $A_{k+1} = 0$. When $i\neq k+1$, we have $(A_j+A_i)|_{R_i,\overline{C_i}} = (A_{k+1}+A_i)|_{R_i,\overline{C_i}} = A_i|_{R_i,\overline{C_i}}$ for $j\in [k], j\neq i$, which implies $A_j |_{R_i, \overline{C_i}} = 0$ and $A_i |_{R_i, \overline{C_i}}$ is invertible. When $i=k+1$, we have $A_1|_{R_{k+1},\overline{C_{k+1}}} = A_2|_{R_{k+1},\overline{C_{k+1}}} = \cdots = A_k|_{R_{k+1},\overline{C_{k+1}}}$ is invertible.

We define the $r$-by-$r$ matrix $B_{k+1}$ as $B_{k+1}|_{R_{k+1}, \overline{C_{k+1}}} = A_1|_{R_{k+1},\overline{C_{k+1}}}$ and the other parts of $B_{k+1}$ are all zero, {\em i.e.}, $B_{k+1}|_{\overline{R_{k+1}}, E} = 0$ and $B_{k+1}|_{E, C_{k+1}} = 0$. Let $B_i = A_i + B_{k+1}$, for $i\in[k]$.
Note that for $i\in [k]$, $B_{k+1}|_{R_{i}, \overline{C_{i}}} = 0$ since otherwise $A_1|_{R_{i}, \overline{C_{i}}} = A_2|_{R_{i}, \overline{C_{i}}} \neq 0$, conflicting the fact $A_j|_{R_i, \overline{C_i}} = 0$ for $i\neq j$.
We now verify the series of matrix $B_i$ satisfy the two properties.

Property 1) holds because $B_{k+1}|_{R_{k+1}, \overline{C_{k+1}}} = A_1|_{R_{k+1},\overline{C_{k+1}}}$ is non-singular and for each $i\in [k]$, $B_i |_{R_i, \overline{C_i}} = A_i |_{R_i, \overline{C_i}} + B_{k+1}|_{R_i, \overline{C_i}} = A_i |_{R_i, \overline{C_i}}$ is non-singular.

Property 2) holds because for the case $i=k+1$, $j\in [k]$, $B_j|_{R_{k+1}, \overline{C_{k+1}}} = A_j |_{R_{k+1}, \overline{C_{k+1}}} + A_1|_{R_{k+1}, \overline{C_{k+1}}} = 0$; for the case $i\in [k]$, $j=k+1$, $B_{k+1}|_{R_{i}, \overline{C_{i}}} = 0$; for the case $i,j\in[k], i\neq j$, $B_{j}|_{R_{i}, \overline{C_{i}}} = A_j |_{R_{i}, \overline{C_{i}}} + B_{k+1}|_{R_{i}, \overline{C_{i}}} = 0$.

\end{IEEEproof}

\subsection{Proof of Theorem 6}
\begin{IEEEproof}[Proof]
For repairing a basic disk $D_i, i \in [k+1]$ with the MDR codes, we can use the row parity to compute block $d_{i,j}$ for $j\in C_i$, which needs $k-1$ XORs. So we only need to consider the case of computing a block $d_{i,j}, j\notin C_i$. Let $d_i|_{\overline{C}_i}$ denote the column vector composed of these blocks.

As shown in the proof of Theorem \ref{thm:optimalRepair}, we use the following equation to calculate $d_i|_{\overline{C_i}}$:
\[
d_{k+2} |_{R_i} = B_i |_{R_i, \overline{C_i}} d_i|_{\overline{C_i}} + \sum_{j=1}^{k+1} B_j |_{R_i, C_i} d_j|_{C_i}
\]
According to the construction of MDR codes, we can see that $B_i|_{C_i,\overline{C}_i} = I$ holds for the initial case and is preserved in the construction. Let $B''_j, j\in [k]$ be the $(k-1, 2^{k-1})$ MDR codes, $R''_j, C''_j$ be the corresponding repair strategy. Let $y_k$ denote the number of XORs to compute $\sum_{j=1}^{k+1} B_j |_{R_i, C_i} d_j|_{C_i}$. If $i=k+1$,
\[
\sum_{j=1}^{k+1} B_j |_{R_i, C_i} d_j|_{C_i} = \sum_{j=1}^{k-1} B''_j d_j|_{C_i} +  B''_k \sum_{j=1}^{k-1} d_j|_{C_i}
\]
which is equivalent to calculating the Q disk of the $(k-1, 2^{k-1})$ MDR codes, which takes $(k-2)2^{k-1}$ XORs according to the analysis in Sec.~\ref{sec:encoding}. If $i<k+1$,
\begin{eqnarray*}
\sum_{j=1}^{k+1} B_j |_{R_i, C_i} d_j|_{C_i} &=& \sum_{j=1}^{k-1} \left[\begin{array}{cc}B''_j |_{R''_i, C''_i}&0\\0 & B''_j |_{R''_i, C''_i}\end{array}\right] d_j|_{C_i} + \left[\begin{array}{cc}B''_k |_{R''_i, C''_i}&0\\0 & B''_k |_{R''_i, C''_i}\end{array}\right] \sum_{j=1}^{k-1} d_j|_{C_i}\\
& & + \left[\begin{array}{cc}0& I_{r/2 \times r/2}\\ 0 & 0\end{array}\right]d_{k}|_{C_i} + \left[\begin{array}{cc}0&0\\I_{r/2 \times r/2} & 0\end{array}\right]d_{k+1}|_{C_i}
\end{eqnarray*}
Note that the first two terms can be calculated recursively with $2y_{k-1}$ XORs. Thus,
$y_k = 2y_{k-1} + 2^{k-1}$. Combining the two cases, we have $y_k = (k-2)2^{k-1}$. As $d_i|_{\overline{C_i}} = d_{k+2} |_{R_i} + \sum_{j=1}^{k+1} B_j |_{R_i, C_i} d_j|_{C_i}$, we can see that $d_i|_{\overline{C_i}}$ can be calculate with $y_k + 2^{k-1} = (k-1)2^{k-1}$ XORs. As there are $2^{k-1}$ blocks in $d_i|_{\overline{C_i}}$, the average number of XORs to repair each block is $k-1$.
\end{IEEEproof}

\end{document}